%% file: main.tex
\documentclass[submission,copyright]{eptcs}
\providecommand{\event}{NCMA 2023} 

\usepackage[utf8]{inputenc}
\usepackage{underscore}
\usepackage[english]{babel}

\usepackage{iftex}

\ifpdf
  \usepackage{underscore}         
  \usepackage[T1]{fontenc}        
\else
  \usepackage{breakurl}           
\fi

\input{defs}

\newcommand{\tok}{\ensuremath{\mathord{\,\raisebox{0.5pt}{\ensuremath{\wr}}\,}}}
\newcommand{\Tk}{\ensuremath{\mathbb{T}}}
\newcommand{\TkEmpty}{\ensuremath{\Tk^{\emptyset}}}

\DeclareMathOperator{\vocab}{vocab}
\DeclareMathOperator{\maxvocab}{maxvocab}
\DeclareMathOperator{\maxtok}{maxtok}

\newtheorem{theorem}{Theorem}
\newtheorem{lemma}{Lemma}
\newtheorem{corollary}{Corollary}
\newtheorem{observation}{Observation}
\theoremstyle{definition}
\newtheorem{definition}{Definition}

\newtheorem{algorithm}{Algorithm}
\theoremstyle{remark}
\newtheorem{example}{Example}
\newtheorem{remark}{Remark}

\title{Ordered Context-Free Grammars Revisited}
\author{
  Brink van der Merwe
  \institute{Department of Computer Science\\Stellenbosch University\\ Stellenbosch, South Africa}
  \email{abvdm@cs.sun.ac.za}
}

\def\authorrunning{B. van der Merwe}
\def\titlerunning{Ordered Context-Free Grammars}

\newcommand{\tokenbox}[1]{%
  \tikz[baseline] {%
    \node[outer sep=0,inner sep=0,anchor=base west] (text) {\hskip.7pt #1\hskip.7pt\vphantom{$\hat{I}$}};%
    \begin{pgfonlayer}{background}%
      \draw[very thick,opacity=0.3,red] (-.1em,-.2em) rectangle (text.north east);%
    \end{pgfonlayer}
  }%
}

\begin{document}

\maketitle

\begin{abstract}
We continue our study of ordered context-free grammars, a grammar formalism that places an order on the parse trees produced by the corresponding context-free grammar. In particular, we simplify our previous definition of a derivation of a string for a given ordered context-free grammar, and  present a parsing algorithm, using shared packed parse forests, with time complexity $O(n^4)$, where $n$ is the length of the input string being parsed.
\end{abstract}


\section{Introduction}
\label{sec:intro}
Ordered context-free grammars (oCFGs), a grammar formalism introduced in~\cite{vdMerwe22}, provides an alternative to parsing expression grammars (PEGs), when requiring an unambiguous grammar formalism. This formalism has much easier to understand matching semantics compared to PEGs, 
 but this comes at the price of much worse parsing time \changeB{complexity. Indeed, the complexity is} $O(n^4)$ compared to linear, where $n$ is the length of the input string being parsed. It should be noted that this is not worse than the adaptive LL(\verb|*|) algorithm, used in the popular parser generator ANTLR~\cite{LLstar}.
Ordered context-free grammars are unambiguous, since we select the least parse tree for a given input string (if possible), based on the order induced on parse trees by the oCFG formalism. 

The matching semantics of oCFGs are more intuitive than PEGs, since an oCFG \change{matches} exactly the same string language as the corresponding context-free grammar (CFG), in contrast to PEGs.
We obtain PEGs from context-free grammars by 
replacing the choice operator, typically denoted by the pipe character `$\ |\ $', by an ordered choice operator, i.e.\ the choice operator becomes non-commutative. The semantics of the ordered choice operator is such that if the first alternative succeeds locally, i.e.\ if the ordered choice lets the current nonterminal consume some substring starting at the current position without regard for the overall match, the second alternative is never attempted. 
Specifying when a rule succeeds locally should be stated with more care -- more on this later in the introduction.
The oCFG formalism also makes use of an ordered choice operator, but the emphasis is on overall instead of local success.
Despite the popularity of PEGs as unambiguous grammar formalism, there are some downsides, for example, proving that a given PEG matches an intended string language is often complicated. As pointed out in~\cite{ComputationalPower}, the influence of PEGs can be illustrated by the fact that despite having been introduced only twenty years ago, the number of PEG-based parser generators \change{exceeds the} number of parser generators based on any other parsing method. 

The unexpected behaviour of PEGs can for example be seen when considering the PEG with $S\rightarrow aSa\, /\, a$ as the only production, describing the regular language $\{a^{2n+1}\mid n\ge 0\}$ when replacing the PEG with the corresponding CFG. Normally, PEGs use the symbol `$\leftarrow$' in productions, and not `$\rightarrow$', although we will deviate from this convention. Next, we explain (informally) why this PEG does not match $a^5$, while matching, for example, $a^3$ and $a^7$. We also discuss this example more formally in Example~\ref{ex:a5}. Given that $S\rightarrow aSa\mid a$ is an example of an unambiguous CFG, we note that the PEG formalism not only makes a CFG unambiguous, but might also reduce the set of strings being matched.
Let's also consider $S\rightarrow aSa \mid a$ as an oCFG. In both the PEG and oCFG case, derivations begin by applying the rule $S\rightarrow aSa$ twice, but the PEG then applies $S\rightarrow aSa$ a \emph{third time,} as it considers this rule as being ``locally successful'', since the right-hand side of the rule consumes the 3rd to the 5th `$a$' (after replacing the $S$ in $aSa$ with $a$). But this will cause the 2nd application of $S\rightarrow aSa$ to fail. 
In comparison to PEGs, the oCFG would select $S\rightarrow a$ as the 3rd rule to apply. This ensures that applying $S\rightarrow aSa$ is successful as the 2nd derivation step, and in this way we obtain a successful oCFG derivation of $a^5$. That is, PEGs select the first locally successful rule, whereas with oCFGs, the first rule which enables overall derivation success, is selected. As stated before, applying a rule $r$ to rewrite a nonterminal $A$, is regarded as locally successful, if by applying $r$, and keeping on rewriting the nonterminals produced by $r$, we obtain a string of terminals which is a prefix of the remainder of the input string. But, in PEGs, the selection later of locally successful rules in a derivation, applied to nonterminals produced by an earlier rule application step, has precedence over the success of the earlier selected rule. Thus, a \changeB{selected rule} might fail in PEGs, since local success preference is given to later applied rules. 

The non-commutativity of the choice operator in PEGs can be seen when changing the above example to $S\rightarrow a\mid aSa$, and noting that in this case only the input string $a$ is matched. In oCFGs, the operator `$|$' is also non-commutative when considering the order on the parse trees produced by an oCFG, but not when only considering the strings being matched. Of course, given that $S\rightarrow aSa\mid a$ is an unambiguous CFG, it makes no difference whether this example grammar is considered as a CFG or as an oCFG.

The oCFG formalism is a natural way to generalize Perl-compatible regular expression (PCRE) matching, to context-free parsing. PCRE matching semantics is used in almost all regular expression matching libraries. See for example~\cite{wild} for a discussion \changeB{on} how real-world regex matching semantics are deeply intertwined with a depth-first backtracking parsing technique. In both PCRE regex matching and oCFGs, ambiguity is removed in perhaps the most natural generic way, i.e.\ when having multiple transition or rule choices, we place and preference on which one should be used, by ordering transitions and rules respectively. 

When considering regular expressions, PEGs correspond to the atomic operator (see~\cite{atomic}), as illustrated in the following example. Consider the regular expression $r:= a^*a$, which we translate into a CFG $G_r$ with productions  $S\rightarrow Aa$ and $A\rightarrow aA\mid \varepsilon$. When using the atomic operator in $r$ to obtain $r'$, with $r':=(\triangleright a^*)a$, we obtain the corresponding PEG $G_{r'}$ with productions $S\rightarrow Aa$ and $A\rightarrow aA\ /\ \varepsilon$. In this case, $(\triangleright a^*)$ consumes locally as many characters as possible, and thus $r'$ and $G'_{r'}$ describe the empty language. In both regexes and grammars, atomic operators and parsing expressions grammars provides respectively improved efficiency in matching or parsing, but at the cost of often difficult to understand or unexpected matching behaviour.

In~\cite{LLstar}, it is pointed out that the parser generator ANTLR, a top-down parser generator developed by Terence Parr, uses the order in which rules are specified, as one way of resolving ambiguities. The parser generator YACC (see~\cite{yacc-rr}) also uses the order of rules to resolve reduce-reduce conflicts.
This observation provides additional motivation for why the oCFG formalism is of interest.
%

\changeB{Strictly speaking, we should rather refer to oCFGs, as ordered parse tree context-free grammars, given, as will be shown in the next section, the order of rules in an oCFG is used to obtain an order on the parse trees. The terminology ``ordered context-free grammars'' is also used in the regulated rewriting community for a related formalism (see for example~\cite{Fris, Lepisto}). In this related formalism, a partial order is placed on the grammar rules, and a rule is not allowed to be applied to a sentential form if a larger rule is also applicable to the sentential form. In contrast to PEGs (or oCFGs), this regulated rewriting formalism determines if a rule is applicable to a sentential form (and that there are no larger applicable rules), and not if a rule is both applicable and succeeds locally (or respectively, guantees overall success).}

In this paper, we simplify the notion of an oCFG derivation in Section~\ref{sec:derivations}, compared to~\cite{vdMerwe22}, by not explicitly modelling backtracking. In Section~\ref{sec:parsing_complexity}, we also consider the complexity of parsing oCFGs, a question not considered before. Results from~\cite{vdMerwe22} required to follow the exposition in this paper, are stated without proof. 
The outline of this paper is as follows. The next two sections provide definitions and elementary results on oCFGs and on PEGs. Then, oCFG derivations are considered, after which we discuss oCFG parsing by using shared packed parse forests. Finally, we present our conclusions and a discussion on envisioned future work. 

\section{Definitions and elementary properties of oCFG}
\label{sec:elemntary}

Next, we define oCFGs. In an oCFG, we order all rules with the same nonterminal on the left-hand side, and then number each of these collections of rules, from one onward.
We consider only the leftmost derivations, and associate a list of integers with each derivation, based on rules used in the derivation, from left to right. Derivations (and parse trees) can thus be compared and ordered, using the lexicographic ordering of the list of integers associated with a derivation.
We also consider a subclass of oCFGs, where for each string $w$ in the language of the grammar, there is a least derivation (and thus parse tree) for $w$. 
Thus, oCFGs extend CFGs in such a way that the strings accepted, and their corresponding parse trees are the same, but we also have an order on the parse trees. 

In the following definition, we define trees, which will mostly be used as parse trees in this paper.

\begin{definition}\label{def:tree}
The set of \emph{ordered, rooted and ranked trees,} over a finite ranked alphabet $\Gamma=\cup_{i=0}^{\infty}\Gamma_i$, denoted by $\mathcal{T}_\Gamma$, where $\Gamma_i$ is the set of alphabet symbols of rank $i$, is defined inductively as follows:
\begin{itemize}
\item if $a\in\Gamma_0$, then $a\in\mathcal{T}_\Gamma$;
\item if $a\in\Gamma_k$ and $t_i\in\mathcal{T}_\Gamma$ for $1\le i\le k$, then $a[t_1,\ldots,t_k]\in\mathcal{T}_\Gamma$. 
\end{itemize}
The height of $t\in\mathcal{T}_\Gamma$, denoted $\height(t)$, is defined inductively as follows. We let $\height(t)=0$ if $t=a\in\Gamma_0$, otherwise, if $t=a[t_1,\ldots,t_k]$, then $\height(t)=1+max(\height(t_1),\ldots,\height(t_k))$.
\end{definition}

Next, we define trees referred to as contexts. Using contexts, we can construct a larger tree by substituting the special symbol $\Box$, by another tree.

\begin{definition}
    \label{defn:context}%
\changeB{Assume $\Box$ is a symbol of rank 0 that is not in the ranked alphabet $\Gamma$. Denote by $\mathcal{C}_\Gamma$ the set of trees over the ranked alphabet $\Gamma\cup\{\Box\}$, where each tree has precisely one leaf node labelled by $\Box$.
A tree in $\mathcal{C}_\Gamma$ is referred to as a \emph{context}.}

For $t\in\mathcal{C}_\Gamma$ and $t'\in \mathcal{C}_\Gamma\cup\mathcal{T}_\Gamma$, denote by $t\llbracket t'\rrbracket\in\mathcal{C}_\Gamma\cup\mathcal{T}_\Gamma$ the tree obtained by replacing the instance of $\Box$ in $t$, by $t'$. 
\end{definition}

Now, we are ready to define ordered context-free grammars, which at this stage, looks the same as CFGs. The way in which we extend CFGs to obtain oCFGs, will become clear once explain how to order parse trees.

\begin{definition} 
An \emph{ordered context-free grammar} $G$ is a tuple $(N,\Sigma,P,S)$, where:
\begin{itemize}
\item[(i)] $N$ is a finite set of nonterminals;
\item[(ii)] $\Sigma$ the input alphabet; 
\item[(iii)] $P$ is the production function and for $A\in N$, we have $P(A)=(r^A_1,\ldots,r^A_{n_A})$, with $r^A_i\in(N\cup\Sigma)^*$;
\item[(iv)] $S \in N$ is the start nonterminal.
\end{itemize}
\end{definition}

When $P(A)=(r^A_1,\ldots,r^A_{n_A})$, we  also use the notation $A\rightarrow r^A_1 \mid \cdots \mid r^A_{n_A}$.
The order of the $r^A_i$, in  $(r^A_1,\ldots,r^A_{n_A})$, will play a role in the order of the parse trees, defined later. In results where order is not important, we will mostly use the terminology CFG, instead of oCFG.

We refer to $A\rightarrow r^A_1 \mid \cdots \mid r^A_{n_A}$ as a production, and to $A\rightarrow r^A_i$, for some $1\le i\le n_A$, as a rule. As is usual in CFGs, we say that for $u,v\in (N\cup\Sigma)^*$, that $u$ directly yields $v$, written as $u\Rightarrow v$, if $u=u_1Au_2$ and $v=u_1r^A_iu_2$, for some $1\le i\le n_A$. Also, we denote by $\Rightarrow^*$ the reflexive transitive closure of $\Rightarrow$, and by $\Rightarrow^+$ the transitive closure of $\Rightarrow$. If $S\Rightarrow^* u$, for $u\in (N\cup\Sigma)^*$, we refer to $u$ as a sentential form.

A ranked alphabet $\Gamma_G$ (which we will use in parse trees) is associated with an oCFG $G$ as follows. 
\change{Denote by $|v|$ the length of a string $v$, with the length of the empty string $\eps$ taken to be $0$.}
We let
$\Sigma\cup\{\eps\}$ be the elements of rank $0$ in $\Gamma_G$, since these will label the leafs of the parse trees. \changeB{If $P(A)=(r^A_1,\ldots,r^A_{n_{\!A}})$, then define $A_i$, for $1\le i\le n_A$, to be a symbol of rank $\max\{1,|r^A_i|\}$ in $\Gamma_G$.} \changeB{We use the symbols with subscripts, $A_i$,} in parse trees to encode the production choice $A\rightarrow r^A_i$. Since $r^A_i$ might be equal to $\eps$, we take the rank of $A_i$ to be $\max\{1,|r^A_i|\}$, since a node in a parse tree labelled by $A_i$, will still have a child labelled by $\eps$ when $r^A_i=\eps$.

For a tree $t$, the notation $y(t)$ is used for the yield of $t$, i.e. the string of non-$\eps$ leaf symbols in $t$, considered left to right. Thus, to obtain $y(t)$, we delete $\eps$ and all symbols of rank greater than zero and also `$[$', `$]$' and `$,$' in $t$. In the special case where all leaf symbols are $\eps$, we define $y(t)$ to be $\eps$ as well.


\begin{definition}
For an oCFG $G$ and string $w\in\Sigma^*$, we define the set of \emph{parse trees of $w$}, denoted by $\mathcal{P}_G(w)$, as all trees over the ranked alphabet $\Gamma_G$, satisfying the following criteria:
\begin{itemize}
\item[(i)] The root is labelled by some $S_i$, $1\le i\le n_{S}$, where $S$ is the start nonterminal of $G$;
\item[(ii)] $y(t)=w$;
\item[(iii)] The children of a node labelled by $A_i$, ignoring subscripts of nonterminals, \change{are} labelled, in order, by the symbols in $r_i^A$. As a special case, when $|r_i^A|=0$, a node labelled by $A_i$ will have a single child leaf labelled by $\eps$.
\end{itemize}
The \emph{string language defined by $G$}, denoted by $\mathcal{L}(G)$, is the set of strings $w$ for which $\mathcal{P}_G(w)\ne\emptyset$.
\end{definition}



By $\mathcal{L}_{\mathcal{T}}(G)$ we denote \change{the} set of parse trees of $G$, which is the set $\bigcup_{w\in\Sigma^*}\mathcal{P}_G(w)$. We modified the usual definition of parse trees to make it possible to directly read off the productions used to obtain the parse tree, by considering the indices of the nonterminal labels used in the parse tree. More precisely, when doing a pre-order traversal of the non-leaf nodes of a parse tree, the integer subscripts of the nonterminals describe uniquely (with the subscript of a nonterminal indicating which rule choice, from a given production, was made for a given nonterminal) the productions used in a left-most derivation to produce the respective parse tree. Since we know that derivations start with the initial nonterminal $S$, it is not required to know both the nonterminals and their respective indices to deduce the productions used, i.e.\ the indices are  sufficient.


\changeB{For $t\in\mathcal{L}_{\mathcal{T}}(G)$, let $n(t)$ denote the sequence of integers obtained by replacing all symbols $A_i$ in the representation of $t$, as used in Definition~\ref{def:tree}, by $i$, and deleting all other symbols (i.e.\ `$[$', `$]$', `,' and terminal leaves) in the representation of $t$.}

\begin{definition}[Total order on parse trees]
A total order $\prec_G$ is defined on $\mathcal{L}_{\mathcal{T}}(G)$ by letting \emph{$t_1\prec_G t_2$} when $n(t_1)$ is smaller than $n(t_2)$ lexicographically.  
\end{definition}


When having unit or empty rules, oCFGs might not have  well-ordered sets of parse trees for each given input string, and since this is relevant to ensure that oCFGs are unambiguous grammar formalisms, we focus on the following two classes of oCFGs.
\begin{definition}
\changeB{
    Let $G$ be any oCFG. 
    \begin{itemize}
    \item We define $G$ to have \emph{least parse trees} or simply \emph{least trees}, if for all strings $w$, $\mathcal{P}_G(w)$ is either empty or has a least parse tree. 
    \item We define $G$ to be \emph{well-ordered}, if for all strings $w$, the set of trees $\mathcal{P}_G(w)$ is well-ordered (i.e.\ every subset of $\mathcal{P}_G(w)$ has a least parse tree). 
    \end{itemize}
}
\end{definition}
An oCFG having least trees is sufficient to turn oCFGs into an unambiguous grammar formalism by for each $w$ selecting the least tree in $\mathcal{P}_G(w)$. The well-ordered property is stronger, but it is decidable as shown in Theorem~\ref{thm:decide-well-ordered}, in contrast to determining if an oCFG has least trees, which is not decidable (see~\cite{vdMerwe22}).

We can use the order $\prec_G$ to define a {\em filter} on the set of parse trees of the oCFG $G$ (see~\cite{Klint} for more on using filters for disambiguation). For a set $A$, denote by $\Pi(A)$ the power set of $A$. Then a  function $\mathcal{F}:\Pi(\mathcal{L}_{\mathcal{T}}(G))\rightarrow \Pi(\mathcal{L}_{\mathcal{T}}(G))$ is a filter, if for $\Phi\in\Pi(\mathcal{L}_{\mathcal{T}}(G))$, we have $\mathcal{F}(\Phi)\subseteq\Phi$. We define the filter  $\mathcal{F}_G$ such that $\mathcal{F}_G(\Phi)$ consists of the trees $t\in\Phi$, such that for no tree $t'\in\Phi$ (with $t'\not=t$), we have $t'\prec_G t$. Then $G$ having least trees is equivalent to the filter $\mathcal{F}_G$ being \emph{complete}, where a filter is complete if it selects one tree from each non-empty set $\mathcal{P}_G(w)$. 

Instead of using the positive natural numbers, i.e.\ a totally ordered set, to index each of the rules in a given production, from $1$ onwards, we can index the rules by a partially ordered set. These indices can then be used in a lexicographic way, to define a partial order on parse trees. In this way, one can support ordered and unordered choice between rules in a production. Again, we obtain a filter on the set of parse trees, as before, but not necessarily a complete filter. More than one filter can of course be used to remove ambiguity, for example in the LR parser YACC, one could have shift-reduce and reduce-reduce conflicts, where shift-reduce conflicts are resolved by preferring shift over reduce, and only reduce-reduce conflicts are resolved by using the order in which rules are specified.

Next, we provide a sufficient condition for a grammar $G$ to be well-ordered. 
In particular, we provide a necessary and sufficient condition so that all strings $w$ will have finitely many parse trees. We in fact give a necessary and sufficient condition for the opposite, i.e.\ a condition to ensure that some strings will have infinitely many parse trees, which can then be negated.  We assume all nonterminals in $G$ are useful. We define a nonterminal $A$ in $G$ to be \emph{useful} if a sentential form can be derived from $S$ containing $A$, and if a string of terminal symbols can be derived when starting from $A$. 
We say a grammar $G$ is \emph{cyclic} if for some nonterminal $A$ in $G$, we have $A\Rightarrow^+ A$, with $\Rightarrow^+$ being the transitive closure of $\Rightarrow$. Being cyclic is a necessary condition for some strings to have infinitely many parse trees, and conversely,  if each nonterminal in $G$ is useful, then $G$ being cyclic is sufficient for some strings $w$ to have infinitely many parse trees. We thus obtain the following result, generalizing Lemma 1 in~\cite{vdMerwe22}. If in an oCFG $G$ we have $A_1\Rightarrow A_2\Rightarrow\ldots\Rightarrow A_n$, for nonterminals $A_1,\ldots,A_n$ where $A_1=A_n$, we say $G$ has a {\em cycle of unit rules}.
\begin{lemma}
    \label{lemma:sufficient-minimal}%
    \changeB{
Let $G$ be a CFG with all nonterminals being useful.
\begin{enumerate}
    \item  If $G$ is not cyclic, then $\mathcal{P}_G(w)$ is finite for all $w\in\Sigma^*$. 
    \item If $G$ is cyclic, then some strings will have infinitely many parse trees. 
    \item If $G$ neither has any $\eps$-rules nor cycles of unit rules, then it is not cyclic.
    \item If $G$ neither has any $\eps$-rules nor cycles of unit rules, then it is well-ordered.
\end{enumerate}
}
\end{lemma}
\begin{proof}
    Observe that the only way a given string can have parse trees of unbounded size (and thus infinitely many parse trees) is if $G$ is cyclic. Also, conversely, if all nonterminals are useful, then when we have nonterminals involved in cycles, these nonterminals must appear in some parse trees, and we can repeat these cycles as many times as we want in parse trees, without changing the strings being parsed.
    \changeB{From these observations we obtain (1) and (2).}
    \changeB{Statement (3) follows from the definition of a grammar being cyclic, and (4) follows from (1), (3), and the observation that finite ordered sets are in fact well-ordered.}
\end{proof}

\change{The previous lemma implies that an oCFG in Chomsky normal form is well-ordered. 
Thus, the class of string languages recognized by well-ordered oCFGs, or oCFGs with least parse trees, is equal to the class of context-free languages.} 

\begin{example}
    \label{ex:ordered-arith}%

In this example, we give a well-ordered oCFG for arithmetic expressions, with parenthesis used as usual to indicate precedence. It is also considered how an equivalent grammar could be specified in the popular parser generator ANTLR (see~\cite{antlr}). We allow addition ($+$), subtraction ($-$), multiplication ($*$), division ($\div$) and exponentiation ($\string^$), and the oCFG is constructed in a way to indicate precedence and associativity of these operators in the parse trees. Left associativity (for $+,-,*,\div$) is encoded as $S\rightarrow S\, P \, S\mid x$, $P\rightarrow +\mid -$, and $S\rightarrow S\, T \, S\mid x$, $T \rightarrow *\mid \div$, and right associativity (for $\string^$) as $S\rightarrow x\mid S \,\string^ S$. To reflect precedence in the parse trees, operators with lower precedence are specified first. Putting these observations together, we obtain the following oCFG: $$S\rightarrow S\, P \, S\mid S\, T \, S\mid x \mid (S) \mid  S\,\string^ S,\ \ P\rightarrow +\mid -,\ \ T \rightarrow *\mid \div$$ 

ANTLR can handle (only) direct left recursion by making use of grammar rewriting, and will by default assume that operators are left associative, unless specified otherwise. In contrast to oCFGs, the choice between left and right associativity can not be enforced by making use of the order in which rules are specified, and the order of the placement of a rule having only a terminal (or terminals) in the right-hand side (for example $S\rightarrow x$), has no influence on the parse tree produced. Also, ANTLR assumes that rules are specified in the reverse order as used in oCFGs. Thus, the ANTLR equivalent of this grammar will be:
$$S\rightarrow\ <\!\textrm{assoc=right}\!> S\,\string^ S\mid (S)\mid S\, T\,  S\mid S\, P\, S\mid x\ ,\ \ P\rightarrow -\mid +\ ,\ \ T \rightarrow \div\mid *$$ \eqed
\end{example}
\begin{observation}
    \label{obs:no-simple-arith}%
    \changeB{The arithmetic operator oCFG in Example~\ref{ex:ordered-arith} does produce the correct (to be defined in the motivation below) least parse trees, but no grammar with single nonterminal does. More broadly, having various required combinations of precedence and associativity will still require significant grammar rewriting to produce a correct abstract syntax tree (AST).}
\end{observation}

\paragraph{Motivation.}
\changeB{
    Intuitively, we are seeking grammars which produce least trees which do not \emph{misrepresent} the priority and associativity of the operators. More precisely, when replacing the rule $S\rightarrow (S)$ with $S\rightarrow y$, and keeping the other rules as is, we want this new oCFG to produce parse trees reflecting the correct priority and associativity of operators. When comparing the oCFG without the rule $S\rightarrow y$, with the new oCFG having this rule, we regard the terminal $y$ in the new oCFG as representing recursively (note, parenthesized subexpressions might themselves contain more parenthesized subexpressions) the parse tree of a parenthesized expression (when considering smallest parse trees). Also, in the new oCFG, we convert the parse trees to ASTs, by replacing $S[SP[+]\,S]$ with $+[S\, S]$, and similarly for $-,*,\div$, and repeating this replacement on the two inner $S$'s in $+[S\, S]$, and also replacing $S[x]$ with $x$ and $S[y]$ by $y$. Also, we replace the $y$'s inductively by the ASTs of the parenthesized subexpressions they represent.  In these ASTs we now do not allow $+$ or $-$ as the right child of a $+$ or $-$ node, and similar for $*$ and $\div$. We also do not allow \string^ as a left child of a node labelled by \string^. Additionally, we do not want $+$ or $-$ nodes below $*$, $\div$ or \string ^ nodes in the AST, and similarly for $*$ and $\div$ nodes below \string^ nodes.}

    The grammar in Example~\ref{ex:ordered-arith} can be shown to be correct by induction. Observe that a least tree will never contain the subtree pattern $S_1[\alpha,\beta,S_1[\gamma_1,\gamma_2,\gamma_3]]$, for any subtrees $\alpha,\beta,\gamma_1,\gamma_2,\gamma_3$, as the tree $S_1[S_1[\alpha,\beta,\gamma_1],\gamma_2,\gamma_3]$ will necessarily be smaller. This establishes the left-associativity of addition and subtraction, and correct associativity for multiplication, division and  exponentiation can be shown similarly. Precedence is obtained by noting that rules for lower priority operators are specified first, and this ensures that they then appear higher up in the parse trees and ASTs. 

    For the second part, observe the role $P$ and $T$ play in the grammar: they make it possible for operators to have the same precedence. That is, $x+x-x+x$ should be parsed as $((x+x)-x)+x$, treating $+$ and $-$ as interchangeable from a syntactic structure perspective. Simply inlining the operators, as in $S \to S\,\mathord{+}\,S \mid S \,\mathord{-} S \mid \cdots \mid x\mid \cdots$, does not work, as $x+x-x+x$ would produce a least tree 
    describing $(x+(x-x))+x$. Reversing $+$ and $-$, similarly, gives an incorrect tree for $x-x+x-x$. Although this is not the only grammar rewriting to consider, 
    we will not provide exhaustively all arguments required. \eqed
    %
\begin{observation}
    From the last paragraph in the motivation of the previous observation, we see that one needs to be cautious when applying some otherwise natural-seeming grammar rewriting. Specifically, replacing $X \to \gamma Y \delta$ and $Y \to \alpha \mid \beta$, by $X\to \gamma \alpha \delta \mid \gamma \beta \delta$, with $\alpha,\beta\in\Sigma^*$, might not preserve the ordering. More precisely, it is not the case that when taking the smallest parse trees when using the original grammar $X \to \gamma Y \delta$, that one can now replace $Y[\alpha]$ and $Y[\beta]$, by $\alpha$ and $\beta$ respectively, and then obtain the smallest parse trees when using the grammar $X\to \gamma \alpha \delta \mid \gamma \beta \delta$.
\end{observation}

The next theorem also appears as Theorem 2 in~\cite{vdMerwe22}, but the proof that follows is significantly more readable and provides more insight, and also specifies the time complexity of deciding if an oCFG is well-ordered. One can regard the argument in the proof as analysing the potential cycles that might appear in the shared parse forests of input strings. If there are no cycles in the parse forest of an input string, then there are only finitely many parse trees for the given string, but if the parse forest contains a cycle that creates smaller trees when followed, there will be an infinite set of decreasing parse trees. Shared packed parse forests are defined and used in Section~\ref{sec:parsing_complexity}, but the proof of the following theorem can be followed without any knowledge about parse forests.
\begin{theorem}
    \label{thm:decide-well-ordered}
    It is decidable, in time $\mathcal{O}(p\,|N|)$, where $p$ is the sum of the lengths of right-hand sides of the productions in $P$, whether an oCFG $G=(N,\Sigma,P,S)$ is well-ordered. 
\end{theorem}
\begin{proof}
\changeB{Since we can determine in time $\mathcal{O}(p\,|N|)$ which nonterminals are useful, and then discard rules involving these, we may assume that all nonterminals in $G$ are useful.} Recall, we refer to a nonterminal in $G$ as being cyclic if $A\Rightarrow^+\!\!A$. Also, we define a rule $A\rightarrow r$ to be cyclic if $A\Rightarrow r\Rightarrow^*\!\!A$. Now, observe that $G$ is well-ordered if and only if all cyclic rules have the highest possible index (i.e.\ appear last) in the production in which they occur, i.e. if $A\rightarrow r^A_1\mid\ldots\mid r^A_{n_A}$, then there is at most one possible cyclic rule amongst the rules $A\rightarrow r^A_i$, and if there is one, it is the rule $A\rightarrow r^A_{n_A}$. 
To see this, first note that if there are no cyclic rules, then $G$ is well-ordered, since then all strings will have only finitely many parse trees. Also, if all cyclic rules appear last, i.e.\ as $A\rightarrow r^A_{n_A}$, then smaller trees are obtained when removing these cycles, and there are only finitely many parse trees, when not using cycles. Next, note that if we have a cyclic rule $A\rightarrow r^A_i$, with $i<n_A$, then $G$ is not well-ordered. This follows from a pumping argument: observe that some parse tree containing a node labelled $A_{n_{\!A}}$\! (at least one exists as $A$ is useful) can in that case be modified into a smaller (under $\prec_G$) parse tree by instead applying the cyclic rule $A\rightarrow r^A_i$ in that position, rather than using a rule from the production for $A$ with a larger index. We then use the cyclic derivation to produce a new $A_{n_{\!A}}$\! lower down in the tree. Iterating this process gives rise to an infinite sequence of smaller trees, violating well-orderedness.

Thus, we can decide whether $G$ is well-ordered by:
\begin{itemize}
    \item[(i)] Computing the nullable nonterminals; the nonterminals in the smallest set $M\subseteq N^*$, such that $A\in M$ if and only if there is a rule $A \to r$ with $r\in M^*$ (i.e.\ the Kleene closure of $M$), and as base case to this inductive definition, we use $M=\emptyset$ and $\eps\in M^*$;
    \item[(ii)] Computing the set of cyclic rules; search rules participating in cycles in the graph induced by having an edge from $A\in N$ to $B\in N$ if there is a rule $A \to r$ with $r=r'Br''$ for $r',r'' \in M^*$;
    \item[(iii)] Checking that cyclic rules only occur last in their respective productions.
\end{itemize}
Suitably implemented, each of these three steps can be done in time $\mathcal{O}(p\,|N|)$, where $p$ is the sum of the lengths of right-hand sides in $P$.
\end{proof}


We conclude this section by providing a bound on the length of derivations producing least oCFG trees, assuming no $\eps$-rules. 
First, we recall a related result for CFGs.
\begin{theorem}[Thm.~1 in~\cite{cfg-deriv-complexity}]
    \label{thm:length-finite}%
    For a CFG $G=(N,\Sigma,P,S)$ with no $\eps$-rules, the length of a \emph{shortest} CFG derivation for $w\in \mathcal{L}(G)$ is at most $(2|w|-1)|N|$.
\end{theorem}

\begin{corollary}
    \label{cor:length-finite} Let $G$ be an oCFG without $\eps$-rules. Then the bound in Theorem~\ref{thm:length-finite} also holds for a CFG derivation of a \emph{least tree} in $G$ (if a least tree exists for the string $w$). 
\end{corollary}
\begin{proof}
    Refer to the proof in~\cite{cfg-deriv-complexity}, and observe that the bound is achieved by eliminating cycles. 
    To see that the result also  applies to \emph{all} oCFGs with no $\eps$-rules, observe that if a least parse tree exists, it cannot ``contain'' a cycle. That is, the least parse tree cannot be such that $t=c\llbracket c'\llbracket c''\rrbracket\rrbracket$, where (i) $c'\not=\Box$, (ii) $c'$ and $c''$ have the same root label, and (iii) $c\llbracket c''\rrbracket$ is also a parse tree for the same string, since then, either $c\llbracket c''\rrbracket$ or $c\llbracket c' \llbracket c' \llbracket c''\rrbracket\rrbracket\rrbracket$ must be smaller. If $c\llbracket c' \llbracket c' \llbracket c''\rrbracket\rrbracket\rrbracket$ is smaller, then we can keep on repeating the context $c'$, and in this way, each time obtain a smaller tree. 
\end{proof}



\begin{remark}\label{rem:height}
If we allow $\eps$-rules in the previous theorem and corollary in the grammar $G$, then we need to 
replace the length of the derivation by the height of a parse tree obtained from a shortest derivation, and also replace the bound $(2|w|-1)|N|$ by:
\begin{equation}\label{bound}
\max\{(2|w|-1)|N|+|N|,|N|\}=\max\{(2|w||N|,|N|\}
\end{equation}
To see this, first note that we may assume that we consider parse trees obtained from leftmost derivations of a CFG, not having any cycles in the derivation. Also, it is enough to obtain a result similar to Theorem~\ref{thm:length-finite}, since from this theorem, we obtain the corresponding corollary. Next, note that a nonterminal is not repeated in any node to leaf path in a parse tree (obtained from a shortest derivation), from a nonterminal deriving $\eps$.  Thus, in particular, the bound given in (\ref{bound}) holds when $w=\eps$. Next, let $G'$ be the grammar obtained from $G$ by applying $\eps$-rule removal (to $G$) in the standard way, i.e.\ we replace a rule of the form $A\rightarrow r$ by all possible rules $A\rightarrow r'$, where $r'$ is obtained from $r$ by deleting some (or none) of the nonterminals in $r$ from which $\eps$ can derived (and we also remove all $\eps$-rules).
Now, consider a parse tree $t$ for a string $w\not=\eps$, when using $G$, 
and remove from $t$ all subtrees deriving $\eps$, to obtain a tree $t'$. 
Thus, $t'$ is a parse tree for $w$ when using $G'$. 
Now use the previous theorem on $G'$ and $w$ (note $G$ and $G'$ have the same number of nonterminals). We obtain the bound in (\ref{bound}) by noting that the length bound on a derivation (in Theorem~\ref{thm:length-finite} applied to $G'$) is a height bound on the corresponding parse tree $t'$ (which gives us the height bound $(2|w|-1)|N|$ on $t'$), and then we add back the subtrees deriving $\eps$ to $t'$ to obtain $t$. Thus, we obtain the height bound $(2|w|-1)|N|+|N|$ for $t$ by adding $|N|$ to the height bound for $t'$. \eqed
\end{remark}

\section{Parsing expression grammars}
\label{sec:pegs}

In this section, we formally introduce parsing expression grammars, following~\cite{ford-pegs}, but restricting what we allow as parsing expressions, and also assuming that the nonterminal $S$ is the starting expression, instead of making use of a general parsing expression as starting expression.


\begin{definition}[Parsing expressions]
A \emph{parsing expression} is a string of the form $e_1/e_2/\ldots/e_n$, with $n\ge 1$, and $e_i\in(N\cup\Sigma)^*$, where $N$ and $\Sigma$ are finite sets of nonterminal and terminal symbols respectively. 
\end{definition}

We refer to ``$/$'' as the prioritized choice operator. 
The set of parsing expressions over $N$ and $\Sigma$ is denoted by $\mathcal{PE}(N,\Sigma)$.

\begin{definition}[Parsing expression grammars]
A \emph{parsing expression grammar} (PEG) is a tuple $G = (N,\Sigma,P,S)$, where $N$ and $\Sigma$ are finite sets of nonterminal and terminal symbols respectively,
$P(A)=(e^A_1,\ldots,e^A_{n_A})$, with $e^A_i\in(N\cup\Sigma)^*$,
is the production function, and $S$ the starting expression. 
\end{definition}

We write $A\rightarrow e^A_1/\ldots/e^A_{n_A}$, if $P(A)=(e^A_1,\ldots,e^A_{n_A})$, i.e.\ we interpret $P$ as being a function from $N$ to
$\mathcal{PE}(N,\Sigma)$.  
\changeB{Note, we do not use the typical convention for PEGs, where $A\leftarrow e^A_1/\ldots/e^A_{n_A}$ denotes $P(A)=(e^A_1,\ldots,e^A_{n_A})$.}
As in the case of oCFGs, if we have $A\rightarrow e^A_1/\ldots/e^A_n$, we refer to $A\rightarrow e^A_i$ and $A\rightarrow e^A_1/\ldots/e^A_n$ as a rule and production respectively. 

\begin{definition}[Matching semantics of PEGs]
For a PEG $G = (N,\Sigma,R,S)$, we define a function $\rightsquigarrow_G:\mathcal{PE}(G)\times \Sigma^*\rightarrow\Sigma^*\cup\{f\}$, where $f\not\in (N\cup\Sigma)$ denotes failure $(\rightsquigarrow_G$ will also be used as an infix operator). If $(e,x)\rightsquigarrow_G y$, with $y\in\Sigma^*$, then parsing succeeds by parsing the prefix $y$ of $x$, while if $(e,x)\rightsquigarrow_G f$, then parsing fails. For $a,b\in\Sigma, a\not=b$, with $e,e_1,e_2\in\mathcal{PE}(G)$, and $x,x_1,x_2,y\in\Sigma^*$, and $o\in\Sigma^*\cup\{f\}$, we define $\rightsquigarrow_G$ inductively as follows.
\begin{itemize}
    \item Empty rule: $(\eps,x)\rightsquigarrow_G\eps$;
    \item Terminal success: $(a,ax)\rightsquigarrow_G a$;
    \item Terminal failure: $(a,bx)\rightsquigarrow_G f$;
    \item Nonterminal rule: if $A\rightarrow e$ and $(e,x)\rightsquigarrow_G
o$, then $(A,x)\rightsquigarrow_G o$;
    \item Sequence success: if $(e_1,x_1 x_2 y)\rightsquigarrow_G x_1$  and $(e_2,x_2y) \rightsquigarrow_G x_2$, then\newline $(e_1e_2,x_1x_2y) \rightsquigarrow_G x_1x_2$;
    \item Sequence failure: if $(e_1, x_1x_2) \rightsquigarrow_G f$, or $(e_1, x_1x_2) \rightsquigarrow_G x_1$ and $(e_2,x_2)\rightsquigarrow_Gf$, then $(e_1e_2,x_1x_2) \rightsquigarrow_G f$;
    \item Alternation case 1: if $(e_1,xy)\rightsquigarrow_G x$, then $(e_1/e_2,xy) \rightsquigarrow_G x$;
    \item Alternation case 2: if $(e_1,x) \rightsquigarrow_G f$ and $(e_2,x) \rightsquigarrow_G o$, then $(e_1/e_2,x) \rightsquigarrow_G o$.
\end{itemize}
\end{definition}

Next, we translate $(S,w)\rightsquigarrow_G w'$ into a deterministic derivation, using derivation steps denoted by $\Rrightarrow_w$, similar to $\Rightarrow$, in the case of CFGs. All derivation steps $\Rrightarrow_w$ for PEGs will also be derivation steps $\Rightarrow$ for the corresponding CFG (where we change prioritized choice, i.e. ``$/$'', into non-deterministic choice, i.e. ``$|$'', to go from a PEG to the corresponding CFG), but not necessarily conversely. We have that  $(S,w)\rightsquigarrow_G w'$ if and only if $S\Rrightarrow_w^+ w'$ ($\Rrightarrow_w^+$ denotes the transitive closure of $\Rrightarrow_w$), and $S\Rrightarrow_w^+ w'$ implies $S\Rightarrow^+ w'$.

\begin{definition}[Derivations, languages  and parse trees defined by PEGs]\label{def:peg}
A \emph{PEG derivation step} $u\Rrightarrow_w v$ (for a PEG $G$), w.r.t. $w\in\Sigma^*$, denotes that it is possible to 
use a rule $A\rightarrow r^A_i$ (in $G$) from the production $A\rightarrow r^A_1/\ldots/r^A_{n_A}$, to replace the left-most nonterminal in the string $u\in(N\cup\Sigma)^*$ (which thus must be an $A$), by $r^A_i$, to produce the string $v$. Also, if $u'$ is the prefix of $u$ in $\Sigma^*$ to the left of $A$, we require that $(u'r_i^A,w)\rightsquigarrow_G u'v'$, for some $v'\in\Sigma^*$, i.e.\ $(u'r_i^A,w)\not\rightsquigarrow_G f$. Additionally (in contrast to left-most derivation steps $u\Rightarrow v$ in CFGs), we require that $(u'r_j^A,w)\rightsquigarrow_G f$, for $j<i$. If $S\Rrightarrow_w^+ w'$, with $w'\in\Sigma^*$, for some $w$, then $w'$ is in the \emph{language defined by $G$}, and the (left-most) rule applications in the steps of this derivation (in order), is  used to construct a \emph{parse tree for $w$}.
\end{definition}

Note that if $S\Rrightarrow_w^+ w'$, then $w'$ is a prefix of $w$, and $S\Rrightarrow_{w''}^+ w'$ for any $w''$ that contains $w'$ as prefix and is a prefix of $w$, in particular, $S\Rrightarrow_{w'}^+ w'$.

\begin{example}\label{ex:a5}
In this example, we consider the PEG discussed in the introduction with production $S\rightarrow aSa/a$, and show that $a^5$ is not accepted. 


We use Definition~\ref{def:peg}. To see that $S\Rrightarrow_{a^5} aSa\Rrightarrow_{a^5} a^3$, we need to show that $( aaSa,a^5)\rightsquigarrow_G f$, which is the case since $(a^2S,a^5)\rightsquigarrow_G a^5$. This follows by verifying that $(S,a^3)\rightsquigarrow_G a^3$.
\eqed
\end{example}

\section{oCFG derivations}\label{sec:derivations}

We define in this section oCFG derivations and also show the close relationship between PEG and oCFG derivations. 
We obtain oCFG derivations by reformulating left-most CFG derivations to be deterministic by selecting the first rule choice, from a given production (for a given nonterminal), in the order they are specified in the oCFG, that will ensure a successful derivation.  
In our setting, derivations will be left-most, but by definition also deterministic, in contrast to how CFG derivations are typically defined.
\changeB{Also, only strings with smallest trees will have finite derivations. This can be seen by using the definition of a derivation, and also from the definition a smallest parse tree.}

Derivations will be done in one of two modes: prefix mode, where parsing a prefix of the input string is regarded as a success, and full mode, where the complete input string must be parsed. 
The situation is similar to typical PCRE-style regular expression matchers, where the matcher can either be forced to determine if a full match is possible, or be asked to return the first prefix match.

Strictly speaking, we should use a symbol other than `$\Rightarrow$' in oCFG derivations, to distinguish between CFG and oCFG derivations, and we should also indicate, for which string a derivation is computed, just as the notation `$\Rrightarrow_w$' used for PEG derivations, but to keep our notation simple, we will still use `$\Rightarrow$' in oCFG derivations.

PEGs with left recursion lead to infinite derivations, without producing a parse tree, in contrast to non-cyclic oCFGs.
This is for example the case with the PEG having the production $S\rightarrow Sa / a$. But in contrast, in non-cyclic oCFGs we have finite derivations.
Various ways of extending PEGs to support left-recursion have been proposed, for example in~\cite{LeftRecursion}, which is used in the Pegen implementation~\cite{pegen}, but these approaches often lead to unexpected parsing results in corner cases, as is pointed out in the section on related work in~\cite{LeftRecursion1}.  The time complexity of parsing also becomes quadratic in the length of the input string being parsed.

Next, we discuss distinctions between parsing with oCFGs, in contrast to when parsing with PEGs. For PEGs, we can memoize the value $\false$, for pairs $(A,i)$, with $A$ a nonterminal and $i$ a position in the input string, if parsing a prefix of the remainder of the input string from $i$, with $A$, is not possible, and recomputing this, is never necessary. Also, for PEGs, if $t_{i,A}$ is the parse tree when using $A$ as root, and starting at a position $i$ in the input string, then if $t_{0,S}$ makes use of $A$ at position $i$, then $t_{0,S}$ will have $t_{i,A}$ as substree, and this subtree will be a parse tree for a prefix of the string starting at position $i$. Thus, for PEGs, we can also memoize parsing related to successful parsing starting from a given position in the input string with a given nonterminal. These memoization observations are not applicable to oCFGs, and they are the main reason why parsing with PEGs (when not having left recursion), can be done in linear time, in contrast to when parsing with oCFGs. Conceptually, we can regard PEGs as ignoring the overall sentential form when making rule selections during derivation steps, and only focussing on producing locally successful parse trees, when starting from a given position with a given nonterminal, with preference given to later subderivations being locally successful.

Next, we note, as one would expect, that oCFG derivations produce least parse trees. 

\begin{theorem} \changeB{The rules in a derivation of a string $w$ with a least tree over an oCFGs $G$, applied in order, in a left-most way, produce the least parse tree of $w$.}
\end{theorem}

\begin{proof}
The result follows directly from the definition of derivations in oCFGs.
\end{proof}

In the next section, we consider the complexity of determining  an oCFG derivation of a complete input string, by making use of the shared packed parse forest for the input string. 

\section{oCFG parsing with shared packed parse forests}\label{sec:parsing_complexity}


In this section, we show how to use shared packed parse forests (SPPFs) to compute oCFG derivations. 
First, we argue why considering  only the case where the complete input string is parsed, is sufficient to also handle parsing in prefix mode. To turn prefix mode into a special case of parsing the complete input string, we note that to simulate prefix mode with full mode, we simply add a new start nonterminal $S'$ with a rule $S'\rightarrow SA$, where $S$ is the old start nonterminal, and $A$ a new nonterminal not used elsewhere in the oCFG productions, and for $A$ we add a production to ensure that $A$ can parse any length input string.

In terms of our presentation of SPPFs, we follow~\cite{bsr} closely.
An SPPF encodes all parse trees of a string $w$, derived from a CFG $G$, in a graph $P$, with the root node labelled by $(S,0,|w|)$, the number of nodes in $P$ at worst cubic in $|w|$, and the height of a path not following cycles, bounded by $O(|w|)$, with the constant determined by $G$. 

To define the SPPF $P$ for a string $w$, derived from the CFG $G$, we first introduce indexed binary derivation trees.
An indexed binary derivation tree (BDT) is constructed from a derivation tree by first introducing intermediate nodes, so that the tree is binarised from the right. Thus, when a node in a parse tree has more than two children, we keep the leftmost child as is, but concatenate the labels of all the other children, to obtain the label of the new right child. 
In contrast to how BDTs (and SPPFs) are typically presented, we binarise from the right, instead of from the left, since this corresponds more closely to how top-down, left-to-right parsing works.
As is usually the case, we add to the labels of nodes, in the BDT,  two integers, $i$ and $j$, which are the left and right positions, in $w$, of the substring at their leaves. Also, if $(x\alpha,i,j)$ is the label of a node $n$ in a BDT, with $|x\alpha|\ge 2$, where $x\in(N\cup\Sigma)$, then the left child of $n$ is labelled by $(x,i,k)$, and the right child by $(\alpha,k,j)$.   
Consider for example the CFG with rules, $S\rightarrow BAa\mid bAa,\ A\rightarrow a,\ B\rightarrow b$, and the input string $baa$. Then we have two BDTs, one in which the root node, labelled by $(S,0,3)$, has a left child $(B,0,1)$, and a right child $(Aa,1,3)$. In the other BDT, $(S,0,3)$ has left child $(b,0,1)$, and right child $(Aa,1,3)$. Nodes in the BDT labelled by $(X,i,j)$, with $X\in(N\cup\Sigma\cup\varepsilon)$, will be referred to as symbol nodes, and those labelled by $(\alpha,i,j)$, with $\alpha\in(N\cup\Sigma)^*$, where $|\alpha|\ge 2$, as intermediate nodes. 

A binarised SPPF is obtained from the set of indexed BDTs for $w$, by taking all nodes from the BDTs of $w$, identifying nodes with the same label, and by adding packed nodes. 
Non-leaf symbols nodes and intermediate nodes have one or more packed node children.
A symbol node $(X,i,j)$, with $X$ a nonterminal, has a rule-packed child $(X\rightarrow x\beta,i,k,j)$, with $x\in (\Sigma\cup N)$ and $\beta\in(\Sigma\cup N)^*$, if:  
\begin{itemize}
\item[(i)] $X\rightarrow x\beta$ is a rule in $G$;
\item[(ii)] There is a symbol node labelled by $(x,i,k)$;
\item[(iii)] Either $\beta\not=\eps$ and there is a symbol or intermediate node labelled $(\beta,k,j)$, or 
$\beta=\eps$ and $k=j$. 
\end{itemize}

The nodes $(x,i,k)$, and $(\beta,k,j)$, if $\beta\not=\eps$, are the children of $(X\rightarrow x\beta,i,k,j)$.
An intermediate node $(x\beta,i,j)$ with $x\in (\Sigma\cup N)$ and $\beta\in(\Sigma\cup N)^*$, where $|\beta|\ge 1$, has an intermediate packed node child labelled $(x\beta,i,k,j)$, if there are nodes labelled $(x,i,k)$ and $(\beta,k,j)$, which are then the children of the intermediate packed node. 

In the example grammar $S\rightarrow BAa\mid bAa,\ A\rightarrow a,\ B\rightarrow b$, with input string $baa$, mentioned above, the root node $(S,0,3)$ in the SPPF, has  two rule-packed nodes, namely $(S\rightarrow BAa, 0,1,3)$ and also $(S\rightarrow bAa,0,1,3)$. The node $(S\rightarrow BAa, 0,1,3)$ has children $(B,0,1)$ and $(Aa,1,3)$, and $(S\rightarrow bAa,0,1,3)$ has a left child $(b,0,1)$, and share its right child, $(Aa,1,3)$, with the node $(S\rightarrow BAa, 0,1,3)$.

The SPPF for the input string $w$, can be constructed in time $O(|w|^3)$, using for example a generalized LL parsing algorithm. Also, the packed nodes on their own uniquely determine the symbol and intermediate nodes and if we only keep them, we have what is known as the binary-subtree representation (BSR) of a SPPF. 

The string $w$ has infinitely many parse trees, precisely when the SPPF has a cycle. 
When no cycle is present, each selection choice of rule-packed nodes, where a unique rule-packed node is selected from all rule packed node children  of a given non-leaf symbol node, and all selected nodes are reachable from the root node, after removing those not selected,  corresponds to a unique parse tree. Once we have made such a selection of rule-packed nodes, we obtain a parse tree, in which the selected rule-packed nodes, arranged in the order obtained by doing a pre-order traversal of the SPPF, provide the rules used in the left-most derivation of the parse tree described by SPPF with selected rule-packed nodes.

If we interpret the grammar $S\rightarrow BAa\mid bAa,\ A\rightarrow a,\ B\rightarrow b$, as an oCFG, then the parse tree for the input string $baa$ is obtained by selecting the rule-packed node $(S\rightarrow BAa, 0,1,3)$ from the SPPF, and discarding $(S\rightarrow bAa,0,1,3)$.

\begin{example}
Consider the CFG $S\rightarrow SS\mid b$, with $bbb$ as input. Then the root node of the SPPF, is $(S,0,3)$, and this node has the rule-packed nodes $(S\rightarrow SS,0,1,3)$ and $(S\rightarrow SS,0,2,3)$ as children, which reflect the fact that we have two parse trees for the input string $bbb$. Note, in this case, when interpreting the CFG as an oCFG, it is not immediately clear that the rule-packed node $(S\rightarrow SS,0,2,3)$ should be selected, and $(S\rightarrow SS,0,1,3)$ discarded, in order to obtain the oCFG parse tree $S_1[S_1[S_2[b], S_2[b]],S_2[b]]$, for $bbb$, from the parse forest.

Next, consider the SPPF for $S\rightarrow SS\mid b\mid\eps$, and input $b$. In this case, the root node $(S,0,1)$ has 
$(S\rightarrow SS,0,0,1)$, $(S\rightarrow SS,0,1,1)$ and $(S\rightarrow b,0,1,1)$, as packed node children. Note that $(S\rightarrow SS,0,0,1)$ has the symbol node children $(S,0,0)$ and $(S,0,1)$, and thus in this case we have a cycle in the SPPF, since $(S,0,0)$ has children $(S\rightarrow SS,0,0,0)$ and $(S\rightarrow\eps,0,0,0)$, and $(S\rightarrow SS,0,0,0)$ has two edges back to $(S,0,0)$. This reflects the fact that $b$ has infinitely many parse trees.
\eqed
\end{example}

\begin{theorem}\label{thm:parsing-complexity}
Assume $G$ is an oCFG. Then an oCFG derivation for a string $w$ can be computed in time $O(|w|^4)$. 
\end{theorem}

\begin{proof}
First, we assume we have no cycles in the corresponding SPPF.
We do a bottom up traversal of the SPPF, labelling along the way a node with the concatenation of indices of rules used in a left-most derivation, of the smallest parse tree below it. Thus, when encountering a symbol node with multiple packed node children, these nodes will have different labels, describing the rules used in a derivation of the  smallest parse tree below them, and then amongst these, we select the packed node, with label being lexicographically the least. The traversal takes cubic time, and comparing two integer labels to find lexicographically the required SPPF node to construct the smallest parse tree (in cases where a node has multiple rule-packed node children), takes time linear in the length of the labels, which is bounded by the height of the SPPF. This provides a $O(h|w|^3)$ complexity bound, with $h$ denoting the height of the SPPF. 
The result, for the case when no cycles are present, now follows by observing that $h$ is of order $O(|w|)$.

Now we  consider the complication caused by the removed back edges. If any of these add cycles to the selected parse tree,  the argument used in the proof of Theorem~\ref{thm:decide-well-ordered} can be applied  to determine if taking any of these will lead to a larger parse tree in the order induced by the oCFG, or will lead to an infinite decreasing sequence of trees, if the cycle is repeatedly taken. 
This will inform us if the selected tree is minimal, or if no minimal parse tree exists for the given input string.
\end{proof}


\begin{example}
In this example, we consider $S\rightarrow SS\mid b$, with input $bbb$. The root node has packed node children $(S\rightarrow SS,0,1,3)$ and  $(S\rightarrow SS,0,2,3)$. 
The node $(S\rightarrow SS,0,1,3)$ is labelled by $12122$, and $(S\rightarrow SS,0,2,3)$ by $11222$, with $1$ encoding the use of $S\rightarrow SS$ and $2$, the use of $S\rightarrow b$, in a left-most derivation. Thus, with $(S\rightarrow SS,0,1,3)$ we associate the parse tree obtained with the left-most derivation $S\Rightarrow SS\Rightarrow bS\Rightarrow bSS\Rightarrow bbS\Rightarrow bbb$. Similarly, with $(S\rightarrow SS,0,2,3)$ we associate the parse tree obtained with the left-most derivation $S\Rightarrow SS\Rightarrow SSS\Rightarrow bSS\Rightarrow bbS\Rightarrow bbb$. Given that $11222$ is lexicographically less than $12122$, we select the parse tree with left-most derivation $S\Rightarrow SS\Rightarrow SSS\Rightarrow bSS\Rightarrow bbS\Rightarrow bbb$. 
\eqed
\end{example}

\section{Conclusions and Future Work}
\label{sec:conclusion}
 We have shown that oCFGs provide a good way to understand the relationship between PEGs and CFGs, and it has more natural matching semantics than PEGs, but this comes at the price of worse parsing complexity. 
 Ordered context-free grammars is a natural way in which to extend PCRE regex matching to an ordered context-free grammar formalism, in which it is possible to talk about the first match or least parse tree. 
 The natural next step is to build an oCFG parsing tool, which will make it possible to analyse the effort involved for grammar writers to use the oCFG grammar formalism rather than some of the other well-known grammar formalisms. This will also make it possible to determine experimentally if oCFG parsing is fast enough for practical use on large grammars. 
 We are also interested in adding lookahead predicates, as used in PEGs~\cite{ford-pegs}, to oCFGs, and to study the properties of oCFGs with these extensions, similarly to how Bryan Ford  investigated PEGs with these extensions (see~\cite{ford-pegs}). Once this is added to oCFGs, it is no longer necessary to distinguish between the two modes of parsing, i.e.\ prefix and full mode, since full mode can be obtained from prefix mode by using a predicate to specify that the part of the input string being parsed by the start nonterminal, should not be followed by any character. 
Future work also includes a thorough study of which disambiguation can be done with oCFGs and which not, and a study of which disambiguation mechanisms are available in popular compiler generators, and their use in sample grammars. We would also like to investigate interesting and useful subclasses of oCFGs for which parsing can be done in much better time complexity than $O(n^4)$, where $n$ is the length of the input string being parsed.


\section*{Acknowledgement}

I would like to thank Martin Berglund for reading various versions of this document, and suggesting improvements.

\bibliographystyle{eptcs}
\bibliography{main}

\end{document}

\section{Introduction}

Many modern NLP systems, for example large language models such as the GPT models which underpin services like ChatGPT~\cite{chatgpt}, operate on a \emph{tokenization} of text. This tokenization defines an alphabet of symbols (in the formal languages sense) which include as many common words and fragments of words as possible. For example, the OpenAI GPT-2 model has an alphabet size (a ``vocabulary'' in their terminology) of~50,257 tokens, which is enough to turn the sentence ``taking a ride on a boat'' into\footnote{We are for the purposes of this example ignoring some details, involving whitespace and the string beginning and end.} the token sequence ``\tokenbox{taking} \tokenbox{a} \tokenbox{ride} \tokenbox{on} \tokenbox{a} \tokenbox{boat}'', as all words are common enough to be in the alphabet, but ``partaking in a nautical excursion aboard a vessel'' is tokenized as ``\tokenbox{part}\tokenbox{aking} \tokenbox{in} \tokenbox{a} \tokenbox{n}\tokenbox{autical} \tokenbox{exc}\tokenbox{ursion} \tokenbox{ab}\tokenbox{oard} \tokenbox{a} \tokenbox{ves}\tokenbox{sel}'', as the words are uncommon, but fragments of the words are common enough. This alphabet is then more semantically rich (e.g.\ unusual plurals often being represented as a base word plus a token `s'), and makes the model more robust to misspellings and other minor transformations of the text~\cite{bpe}.

One common way of performing this tokenization is by byte pair encoding~\cite{bpe} (BPE), used by OpenAI GPT models, and e.g.\ the 2023 GPT-SW3 model~\cite{gpt-sw3-tokenizer}, which uses the Google tokenizer implementation SentencePiece~\cite{sentencepiece}. BPE operates similar to a compression technique, with a dictionary of token merges constructed greedily maximizing the number of tokens that get merged in a training set (see Remark~\ref{remark:bpe-construction}).

The way the tokenization procedure is defined and implemented in common tools\cite{huggingface-gpt-2-py,sentencepiece} is essentially \emph{global:} the highest priority rule that can be applied to the text is, no matter where in the text this application would happen. This is not usually a problem, as the text is \emph{pretokenized} by splitting it on whitespace before applying the main tokenization procedure. This is not always possible or desirable however:

\noindent
\emph{Notes below.}

For some local pride I figure we can headline the GPT-SW3 use of the sentencepiece BPE tokenizer~\cite{gpt-sw3-tokenizer}.

\begin{itemize}
\item tokenization used all over for ML NLP
\item BPE a popular kind, used e.g. by OpenAI models (e.g. GPT-2)
\item for some languages pretokenization (i.e.\ cutting on whitespace) is not possible, making it performance sensitive -  BPE often relies on a pre-tokenizer that splits the training data into words. Pretokenization can be as simple as space tokenization, e.g. GPT-2, Roberta - which is obviously problematic with non-whitespace languages such as Thai. 
\item incrementally updating tokenizations of huge texts? small changes in underlying text will often give small updates to the tokenizations, but not always. incremental operations: 
\begin{itemize}
  \item given tokenized string $s=uv$ what is the tokenization of $u$ and $v$?
  \item given tokenized strings $u$ and $v$ what is the tokenization of $uv$?
\end{itemize}
\item denial of service attack risks on incremental systems? (even pretokenized this is possible; just make up a ridiculously long word, e.g. a document with whitespace dropped)
\item e.g. huggingface and sentencepiece use heaps for $n\log n$ish (?) tokenization performance, but seems perfectly implementable on a functional string-to-string transducer.
\item huggingface: if pair $x$ is applied then apply to all possible $x$ at once; whereas sentencepiece: apply top available pair. Only differs when degenerate.
\item argue for transducer: can do streaming tokenization
\item could also mention the unspeakable words saga, that turned out to be Reddit usernames
\end{itemize}

\section{Definitions}

An alphabet $\sig$ is a finite set of symbols. Let $\sig^*$ denote all strings over the alphabet $\sig$, including the empty string $\eps$, and $\sig^+=\sig^*\setminus \{\eps\}$. A sequence of non-empty strings is a \emph{tokenization}, e.g.\ for $u_1,\ldots,u_n\in \sig^+$, we denote this $u_1 \tok \cdots \tok u_n$. By $\sig^{\tok}$ we denote the set of all tokenizations constructed from strings from $\sig^*$. We refer to those strings as the \emph{tokens}. 
Let $\pi : \sig^{\tok} \to \sig^*$ be the concatenation of the strings in a tokenization, e.g.\ $\pi(u_1\tok \cdots \tok u_n)=u_1\cdots u_n\in \sig^*$. As a special case, we let $\pi$ applied to the tokenization with $n=0$, be the empty string. 
When $\pi(u_1\tok \cdots \tok u_n)=w\in\Sigma^*$, we say that \emph{$u_1\tok \cdots \tok u_n$ is a tokenization of $w$}.
For $\tau = u_1 \tok \cdots \tok u_n$, we denote by $|\tau|$ the integer $n$. Thus, $|\tau|=0$ if and only if $\tau$ is the empty tokenization.
Also, for $u\in\Sigma^*$, we let $|u|$ denote the length of $u$, i.e.\ the number of symbols from $\Sigma$ in the string $u$. In will be clear from the context and notation when $\tau$ denotes a tokenization with $|\tau|=1$, i.e.\ a tokenization of length one instead of a string of length one, given that in this case, $\tau$ could be interpreted as either. In fact, given our notational conventions discussed below, the symbol $\tau$ will always represent a tokenization.

In addition to using $|\tau|$ and $|u|$ for the length of a tokenization $\tau$ and length of a string $u$ respectively, we use $|S|$ to denote the cardinality of a (finite) set $S$. 

Differentiating between strings and tokenizations becomes important as we continue, so we adopt some conventions. Let $\sig$ denote the alphabet whenever not otherwise specified. When giving examples, we always use $\sig=\{a,b,c,\ldots\}$. Furthermore, we always let $\alpha,\beta,\gamma$ be variables denoting symbols from the alphabet, $u,v,w$ denote strings, and $\tau,\phi$ denote tokenizations, including all sub-/superscripted variants of each. In other words, $u,v,w\in \sig^*$ and $\tau,\phi,\psi\in \sig^{\tok}$, and for that matter $u_1,u_2\in \sig^*$, $\tau' \in \sig^{\tok}$, etc. As such we may write e.g.\ $u \tok \tau \tok v = \phi$ to mean that $\tau$ is a tokenization where the first token is $u$, the last is $v$, and the intervening tokens form the tokenization $\phi$, so $|\tau|=|\phi|+2$.

Next, we define a byte pair dictionary, which will be used to restrict the set of possible tokenizations for a given string $w$. We will only use $D$ and its sub-/superscripted variants to denote a dictionary.

\begin{definition}
    A byte pair \emph{dictionary $D=[u_1\tok v_1,\ldots,u_n\tok v_n]$} of length $|D|=n$ is a sequence of tokenizations $u_1\tok v_1,\ldots,u_n\tok v_n$, with each tokenization $u_i\tok v_i$ being of length 2. We call each $u_i \tok v_i$ a \emph{rule} and say that (a rule) $u_i \tok v_i$ has \emph{higher priority} than $u_j \tok v_j$, when $i<j$.
\end{definition}

We write $\tau \Rightarrow^D \tau'$ if $\tau=\phi \tok u \tok v \tok \phi'$ and $\tau'=\phi \tok uv \tok \phi'$, for some $u\tok v \in D$.
The dictionary $D$ will always be clear from the context, thus we omit the superscript on $\Rightarrow$. We let $\Rightarrow^+$ and $\Rightarrow^*$ denote the transitive and reflexive transitive closure of $\Rightarrow$, respectively. Also, for $w=\alpha_1\cdots\alpha_n$, we denote by $\TkEmpty(w)$ the tokenization $\alpha_1\tok\cdots\tok \alpha_n$. 

Next, we transfer terminology used in derivations over context-free grammars, to our setting. For $w=\alpha_1\ldots\alpha_n$, we begin a derivation for a tokenization with $\TkEmpty(w)$, although a more complicated pretokenizer step could certainly also be of interest, but not considered in this paper. Whereas in the case of context-free grammars, a derivation step consists of applying a grammar rule by replacing the non-terminal on the left-hand side of a rule, by its right-hand side, in our setting, a derivation step is of the form $\phi\tok u_i\tok v_i\tok\phi'\Rightarrow \phi\tok u_iv_i\tok\phi'$, for $u_i\tok v_i$ in $D$. A derivation terminates when no further rules from $D$ can be applied. Also, just as in the case of context-free derivations, we talk about sentential forms when referring to intermediate tokenizations.

The definition of a base tokenizer on $D$, which ignores the priority of rules in $D$, is as follows.

\begin{definition}
  \label{defn:tokenization}%
  For $D=[u_1 \tok v_1, \ldots, u_m \tok v_m]$ and $w=\alpha_1 \cdots \alpha_n$, we obtain the \emph{base tokenizations of $w$ by $D$,} denoted as $\Tk^D_{\textrm{base}}(w) \subset \sig^{\tok}$, as follows. 
We have $\tau_p \in \Tk^D_{\textrm{base}}(w)$ if: 
  \begin{itemize}
  \item $\tau_0=\TkEmpty(w)$,
    \item $\tau_0 \Rightarrow \cdots \Rightarrow \tau_p$,
    \item there exists no $\tau_{p+1}$ such that $\tau_p \Rightarrow \tau_{p+1}$.
  \end{itemize}
\end{definition}

That is, $\Tk^D_{\textrm{base}}(w)$ are the tokenizations of $w$ which can be achieved by applying rules to an initial tokenization $\TkEmpty(w)$ where all symbols in $w$ are their own token, until a point where no further rules can be applied. 
Thus, to obtain one of the possible base tokenizations, we select non-deterministically a rule from $D$ that can be applied to the current tokenization, until the set of rules that could be applied, is empty.
Given that $|\Tk^D_{\textrm{base}}|\ge1$, since there is non-deterministic choice in selecting the next applicable rule, a SentencePiece tokenizer~\cite{sentencepiece} is defined in such a way to remove ambiguity from the base tokenizer. We will (in a somewhat biased way) refer to this tokenization as the correct tokenization.

\begin{definition}
  \label{defn:tokenization-correct}%
  The \emph{SentencePiece} tokenization of $w$, denoted $\Tk^D(w)$, also referred to as the \emph{correct tokenization} of $w$, is $\tau_n$, with $\tau_n\in \Tk^D_{\textrm{base}}(w)$,  where:
  \begin{itemize}
  \item $\tau_0=\TkEmpty(w)$;
  \item $\tau_0 \Rightarrow \cdots \Rightarrow \tau_n$, and for $0\le i < n$, we pick the decomposition $\tau_i=\phi \tok u \tok v \tok \phi'$, to obtain $\tau_{i+1}=\phi \tok uv \tok \phi'$, in such a way that:
  \begin{itemize}
    \item $u\tok v$ is the highest priority rule in $D$ for which such a decomposition exists;
    \item among the remaining decompositions, we pick the one which minimizes $|\phi|$;
  \end{itemize} 
  \item no further rules apply to $\tau_n$.
  \end{itemize}
\end{definition}

Observe that $\Tk^D(w)$ always exists, and is obtained, intuitively, as follows.
Whenever it is possible to apply a rule from the dictionary $D$ to merge some tokens in the interim tokenization, merge the \emph{highest-priority} rule that occurs, and merge the left-most such pair if multiple occurrences exist. Note that this selects a unique tokenization, for each string $w$, from the set $\Tk^D_{base}(w)$.

\begin{example}
  Take the dictionary $D=[a\tok b, a \tok bc, b\tok c, ab \tok c]$, then the correct tokenization of the string $abcbcab$ is $abc \tok bc \tok ab$, using the following steps:
  \begin{itemize}
    \item Initially, $\tau_0 = a\tok b\tok c \tok b \tok c \tok a \tok b$, $a\tok b\in D$ applies to the leftmost $a \tok b$, producing $\tau_1 = ab \tok c \tok b \tok c \tok a \tok b$.
    \item The rule $a\tok b$ still applies, now to the last two tokens, producing $\tau_2 = ab \tok c \tok b \tok c \tok ab$. The rule $a \tok b$ now no longer applies anywhere.
    \item The next rule $a \tok bc$ does not apply, as there is no token $bc$, but $b \tok c$ does apply, producing $\tau_3 = ab \tok c \tok bc \tok ab$.
    \item The first rule that can now be applied is the rule $ab \tok c$, producing $\tau_4 = abc \tok bc \tok ab$. Now, no further rules apply, so $abc \tok bc \tok ab$ is the correct (SentencePiece) tokenization.
  \end{itemize}
  It is interesting to observe that the rule $a \tok bc$ has the second-highest priority in the dictionary, but was never applied. Also, when using the base tokenizer, $a \tok bc$ can certainly be applied to the string $abcbcab$. Observe that applying $a\tok bc$ would produce the token $abc$, but tokenizing the string $abc$ takes the steps $a\tok b \tok c \Rightarrow ab \tok c \Rightarrow abc$. Corollary~\ref{cor:useless} states that a rule $u_i\tok v_i$ is \emph{useful}, i.e.\ gets applied in the tokenization of some string, if and only if it gets applied when tokenizing the string $u_iv_i$. 
\end{example}

\begin{example}
  \label{ex:infinite-ripple}%
  Take the dictionary $D=[c\tok ab, ab \tok c, a \tok b]$. Then the correct tokenization of $abcabcabcabc$ is $abc \tok abc \tok abc \tok abc$. Notice that this tokenization is achieved left to right. After five steps, we have the tokenization $abc \tok abc \tok ab \tok c \tok a \tok b \tok c$. Contrast this to tokenizing the string $bcabcabcabc$ (i.e.\ we delete the initial $a$) which tokenizes as $b\tok cab \tok cab \tok cab\tok c$, or $cabcabcabcabc$ (i.e.\ adding an initial $c$), which leads to $cab \tok cab \tok cab \tok cab \tok c$, where we need to move left to apply $c\tok ab$, after having applied $a \tok b$.
\end{example}

Interestingly enough, not all tokenization libraries modify the base tokenizer in the same way in order to eliminate ambiguity of tokenization. Let us consider the Python implementation of the GPT-2 tokenizer offered by HuggingFace~\cite{huggingface-gpt-2-py}, which removes ambiguity from the base tokenizer, as follows.

\begin{definition}
  \label{defn:tokenization-hf}%
  The \emph{HuggingFace tokenization} $\tau_n$ of $w$, which we denote by $\Tk^D_{\textrm{hf}}(w)\in \Tk^D_{\textrm{base}}(w)$, is defined as follows, where $\tau_0 = \TkEmpty(w)$. In $\tau_0 \Rightarrow^+ \cdots \Rightarrow^+ \tau_n$, for $0\le i < n$, we select a decomposition $\tau_i=\phi \tok u \tok v \tok \phi'$, such that $u\tok v\in D$ is the highest priority rule applicable to $\tau_i$, and then apply $u\tok v$ from left to right, until it is no longer applicable, in order to obtain $\tau_{i+1}$.
\end{definition}

That is, in both Definitions~\ref{defn:tokenization-correct} and~\ref{defn:tokenization-hf} we rewrite $\tau_i$ by picking the highest-priority applicable rule and applying it at the left-most possible position. However, the SentencePiece semantics picks the highest-priority applicable rule in \emph{every} step, where HuggingFace picks a rule and uses it until it becomes inapplicable. This does create a formal difference in semantics, but as we will see they differ only in cases which may be considered degenerate given the way dictionaries are usually constructed.

\begin{example}
  Take the dictionary $D=[ab\tok a, a\tok b]$ and consider the tokenization of $w=abababab$, which has $\Tk^D(w)=aba \tok b \tok aba \tok b$, but $\Tk^D_{\textrm{hf}}(w)=ab \tok ab\tok ab \tok ab$.
\end{example}

\begin{remark}
  \label{remark:bpe-construction}%
The reason why the $D$ in the previous example is regarded to be degenerate or improper, is that a byte pair dictionary is typically produced~\cite{bpe} by taking a training corpus, initially tokenizing it symbol by symbol, and then iteratively adding the most common token pair to the dictionary, tokenizing, and repeating. For example, in the training corpus $a\tok b \tok c\tok a\tok b\tok c\tok a\tok c\tok a$, the most common token pair is $c\tok a$, which is inserted as the first rule in the dictionary. Then, we continue using $a\tok b \tok ca \tok b \tok ca \tok ca$. Now, the most common pair is $b\tok ca$, and this rule is added to the dictionary, which now consists of the rules $[c\tok a, b \tok ca]$. The new dictionary now produces $a\tok bca \tok bca \tok ca$, as tokenization, and so on. Observe that when constructing a dictionary in this way, a rule $u_i\tok v_i$ cannot have higher priority than the rules needed to produce $u_i$ and $v_i$, a property formalized in the next definition. 
Also, when tokenizing the training corpus with the dictionary obtained through training, using HuggingFace semantics, the tokenization obtained will be the tokenization of the training corpus at the end of the training process, and each rule in the dictionary will be used in this tokenization.
\end{remark}
\begin{definition}
  A dictionary $D=[u_1\tok v_1,\ldots u_i\tok v_i,\ldots,u_j\tok v_j,\ldots,u_n\tok v_n]$ is \emph{proper} if for each $j$ with $|u_j|>1$, there exits an $i<j$ such that $u_j=u_{i}v_{i}$, and similarly, for each $j'$ with $|v_{j'}|>1$, there exists an $i'<j'$ such that $v_{j'}=u_{i'}v_{i'}$.
\end{definition}

Note that each rule in a proper dictionary might not be useful. Consider for example the dictionary $D=[b\tok c,a\tok b,c\tok d, ab\tok cd]$. Then $D$ is proper, but $ab\tok cd$ is not useful. This can be seen by noting that when $\TkEmpty(w)$ contains $a\tok b\tok c \tok d$, then the first rule gets applied to produce $a\tok bc \tok d$, and thus it is not possible to apply the rules $a\tok b$ and $c\tok d$, in order so that $ab\tok cd$ could be applied. But note that if $D$ is constructed from a training corpus, then $D$ is proper and each rule in $D$ is useful. \brinkcomment{Should we worry here about SentencePiece vs HuggingFace, when it comes to useful - of course, the next lemma shows for proper we don't have to worry.}

With the additional assumption that the dictionary is proper, the SentencePiece and HuggingFace tokenizers turn out to have equivalent semantics. This should to some extent be expected, as they are intended to achieve the same results, the HuggingFace approach effectively being a small simplification.
\begin{lemma}
  \label{lemma:proper-unifies-semantics}%
  If $D$ is proper we have $\Tk^D(w)=\Tk^D_{\textrm{hf}}(w)$ for all $w$.
\end{lemma}
\begin{proof}
  By contradiction, assume that some $w$ has $\Tk^D(w)\ne \Tk^D_{\textrm{hf}}(w)$. Let $\tau_0 \Rightarrow \cdots \Rightarrow \tau_n$ be the tokenization steps taken by $\Tk^D(w)$, and $\phi_0 \Rightarrow \cdots \Rightarrow \phi_m$ the tokenization steps taken by $\Tk^D_{\textrm{hf}}$. Let $i$ be the smallest index such that $\tau_{i} \ne \phi_{i}$, such an $i$ must exist, as by Definition~\ref{defn:tokenization} we cannot have $\tau_n \Rightarrow \tau_{n+1}$ or $\phi_m \Rightarrow \phi_{m+1}$ for any $\tau_{n+1}$ of $\phi_{m+1}$, so one sequence cannot be a prefix of the other. We also have $i\ge 2$, as the semantics differ only in that Definition~\ref{defn:tokenization-hf} prefers repeating the previous rule over the highest priority one, but this difference can only be exhibited when there is a previous step to repeat.

  We then have $\tau_{i-2}=\phi_{i-2}$, $\tau_{i-1}=\phi_{i-1}$ and $\tau_i\ne \phi_i$. Let $r_0$ be the rule applied going $\tau_{i-2}\Rightarrow \tau_{i-1}$ (and $\phi_{i-2}\Rightarrow \phi_{i-1}$ as they are equal), $r_1$ the rule going $\tau_{i-1} \Rightarrow \tau_i$, and $r_2$ the rule going $\tau_{i-1} \Rightarrow \phi_i$. That is, with some abuse of notation, the following situation:
  \begin{center}
    \begin{tikzpicture}[node distance=0pt]
      \node (a) {$\cdots \Rightarrow (\tau_{i-2}=\phi_{i-2}) \xRightarrow{r_0} (\tau_{i-1}=\phi_{i-1})$};
      \node[right=of a,yshift=6pt,rotate=15] {$\xRightarrow{r_1} \tau_i \Rightarrow \cdots$};
      \node[right=of a,yshift=-6pt,rotate=-15] {$\xRightarrow{r_2} \phi_i \Rightarrow \cdots$};

    \end{tikzpicture}
  \end{center}
  Then we know that $r_1$ has higher priority than $r_2$, since they must differ, and Definition~\ref{defn:tokenization-correct} (SentencePiece) always picks the highest priority rule applicable. The only possible reason for them to differ is that $r_0=r_2$, i.e.\  the Definition~\ref{defn:tokenization-hf} (HuggingFace) semantics prioritized using the same rule as the previous step. However, as they agree in the previous step this means Definition~\ref{defn:tokenization-correct} semantics picked $r_0$ in that step even though $r_1$ has higher priority than $r_0$, which must mean that $r_1$ was \emph{not} applicable in the previous step. This leads to a contradiction, as applying $r_0$ must then have created a token which made $r_1$ applicable, which with $r_0$ lower priority contradicts $D$ being a proper dictionary. As such, our assumption was wrong and $\Tk^D_{\textrm{hf}}(w)=\Tk^D(w)$ by necessity. 
\end{proof}

With this result in hand, it becomes less relevant to differentiate between the two semantics whenever considering only proper dictionaries.

\begin{remark}
  It can be decided whether $D$ is proper in time $\mathcal{O}(\|D\|^2)$, where $\|D\|=\sum \{|uv| \mid u\tok v \in D\}$. 
  For each $u\tok v\in D$, determine all $u'\tok v'$ such that $uv$ is a substring of $u'$ or $v'$, and for all such rules  $u'\tok v'$, verify that $u'\tok v'$ has lower priority than $u\tok v$.
\end{remark}


Next, we investigate if and when we can use the tokenization of a substring of $w$, in order to construct the tokenization of $w$. First, we consider the following example. Let $u,v$ be strings with $\Tk^D(u)=\tau_1 \tok \tau_2$ and $D(v)=\phi_1 \tok \phi_2$. Then it is not necessarily the case that $\tau_1 \tok \phi_2 = \Tk^D(\pi(\tau_1\tok \phi_2))$ or that $\phi_1 \tok \tau_2 = \Tk^D(\pi(\phi_1\tok \tau_2))$. An easy counterexample is obtained by letting $D=[a \tok a, b\tok b]$, $\tau_1 = a$, $\tau_2=b$, $\phi_1=b$, and $\phi_2=a$. Observe that indeed $\Tk^D(ab)=a\tok b= \tau_1 \tok \tau_2$, and $D(ba)=b\tok a=\phi_1 \tok \phi_2$, \emph{however} $\tau_1\tok \phi_2=a\tok a\ne \Tk^D(aa)$ and $\phi_1 \tok \tau_2 = b\tok b \ne \Tk^D(bb)$. 
This shows that tokenizations can not be decomposed and then again glued together in arbitrary ways. 
However, deriving the tokenization of substrings, given the final tokenization, is sometimes possible, as shown in the following lemma.


\begin{lemma}
  \label{lemma:robust-tok}%
  For a dictionary $D$ and string $w$ such that $\Tk^D(w)=\tau_1 \tok \ldots \tok \tau_k$, it holds that $\Tk^D(\pi(\tau_i))=\tau_i$. The result also holds for HuggingFace semantics. In general, when using SentencePiece  or HuggingFace semantics, if $\phi_1\tok\ldots\tok\phi_k\Rightarrow^*\phi'_1\tok\ldots\tok\phi'_k$, then if $\pi(\phi_i)=\pi(\phi'_i)$ for all $i$, we have that $\phi_i\Rightarrow^*\phi'_i$ for all $i$. 
\end{lemma}
\begin{proof}
  Let $\phi_{1,1} \tok \ldots \tok \phi_{k,1} \Rightarrow \cdots \Rightarrow \phi_{1,n} \tok \ldots \tok \phi_{k,n}$ be the steps taken by the procedure in Definition~\ref{defn:tokenization}, such that $\phi_{i,1}=\phi_i$ and $\phi_{i,n}=\phi'_i$ for all $i$, and $\pi(\phi_{i,j})=\pi(\phi_{i,j'})$ for all $i,j$ and $j'$.
  Removing all duplicates from $\phi_{i,1},\ldots,\phi_{i,n}$ produces the sequence of steps taken by SentencePiece derivations, i.e.\ 
  in each step we apply the highest-priority rule from $D$ as left-most as possible, in the case of SentencePiece semantics, and we apply the highest-priority rule  as many times as possible, in the case of HuffingFace semantics.

  For the first statement, take $\phi_i=\TkEmpty(\tau_i)$ and $\phi'_i=\tau_i$.
\end{proof}




\begin{corollary}\label{cor:useless}
For SentencePiece or HuggingFace semantics, a rule in a dictionary $D$ is useful if and only if it gets applied when tokenizing the string it produces.  
\end{corollary}
    
\begin{proof}
The ``if'' part follows directly from the defintion of useful, so for the converse, assume
we have a derivation $\alpha_1\tok\cdots\tok\alpha_n\Rightarrow \cdots \Rightarrow \phi \tok u \tok v \tok \phi' \Rightarrow \phi\tok uv\tok \phi'$. But then the previous lemma implies that $\TkEmpty(uv)\Rightarrow^* u\tok v\Rightarrow uv$.
\end{proof}

The main purpose of Lemma~\ref{lemma:robust-tok} is that it makes it possible to do tokenizations in streaming and incremental ways. The following corollary may illustrate this purpose more clearly.
\begin{corollary}
  \label{cor:bordered-stays-same}%
Assume $\Tk^D(v)=\phi_1\tok\ldots\tok\phi_k$ and $\Tk^D(w)=\psi_1\tok\ldots\tok\psi_l$ with $\pi(\phi_{i'})=\pi(\psi_{j'})$ for some $i'$ and $j'$, then $\phi_{i'}=\psi_{j'}$. The same result holds for HuggingFace semantics. 
\end{corollary}
\begin{proof}
Follows from Lemma~\ref{lemma:robust-tok}, since by Lemma~\ref{lemma:robust-tok}, the tokenizations $\Tk^D(\pi(\phi_{i'}))$ and  $\Tk^D(\pi(\phi_{j'}))$ are completely determined by $\pi(\phi_{i'})=\pi(\psi_{j'})$ and are not influenced by other $\phi_i$ or  $\psi_i$.
\end{proof}

\section{Tokenizing Online With Finite Lookahead}

In this section, we investigate the following question. When we tokenize a string, in a streaming fashion, how long is the suffix that we need to tokenize again, when we resume tokenization, if we make no assumptions about the input string being tokenized. We refer to this constant as the lookahead constant of a dictionary $D$, and denote it by $l(D)$. One of the first thing to show is that this constant is indeed finite.

\begin{lemma}
  \label{lemma:hf-priority-order}%
  For any proper dictionary $D$ and $\Tk^D_{\textrm{hf}}(w)=\tau$, letting $r_1, \ldots, r_n$ be the sequence of rules applied to produce $\tau$ according to Definition~\ref{defn:tokenization-hf}, it must be the case that $r_1, \ldots, r_n$ are in order of (not necessarily strictly) decreasing priority. 
\end{lemma}
\begin{proof}
  By contradiction, assume that there exists some $r_i$ such that $r_{i+1}$ is higher priority than $r_i$. This means that the tokenization after the first $i$ steps contains the pair $r_{i+1}= u \tok v$, but this pair cannot have been created by the applications of $r_i$, as it is lower priority than $r_{i+1}$ and $D$ is proper, and it also cannot have existed in the tokenization when $r_i$ was picked as the rule to next apply, as that goes against how rules are picked in Definition~\ref{defn:tokenization-hf}. As such, our assumption was wrong, and $r_i$ is higher priority or equal to $r_{i+1}$. 
\end{proof}

This lemma will be used to perform HuggingFace tokenization in a straightforward inductive way, where for a proper dictionary $D=[u_1\tok v_1,\ldots, u_n\tok v_n]$ we can tokenize a string by first applying rule $u_1 \tok v_1$ as many times as possible, then the rule $u_2\tok v_2$  as many times as possible, and so on. We will use the term \emph{refinement} for the way a tokenization is developed in this way, i.e.\ a \emph{tokenization $\tau$ is a refinement of $\tau'$}, if  $\tau'$ can be obtained from $\tau$ by applying $\pi$ to some of the subtokenizations in $\tau$. For example, $\phi\tok\phi'\tok\phi''$ is a refinement of $\phi\tok\pi(\phi')\tok\phi''$ and also of $\pi(\phi)\tok\pi(\phi')\tok\phi''$. We can also consider the opposite notion of a tokenization $\tau'$ being \emph{coarser} than $\tau$ , if $\tau$ is a refinement of $\tau'$. Thus, tokenizations are obtained, in general, by using $D$ to obtain courser and courser tokenizations.
Also, recall that $|D|$ denotes the number of rules in $D$.


\begin{theorem}
  \label{thm:left-ripple-bound}%
  Let $D$ be proper, $|\tau|\ge |D|$ and $\Tk^D_{\textrm{hf}}(\pi(\phi)\pi(\tau))=\phi\tok \tau$. Then $\Tk^D_{\textrm{hf}}(\pi(\phi \tok \tau) \pi(\psi)) = \phi \tok \psi'$, for some tokenization $\psi'$. 
\end{theorem}

\begin{proof}
Let $D=[u_1 \tok v_1, \ldots, u_n \tok v_n]$ and $D_i$ be the $i$-length prefix of $D$, i.e.\ $D_i=[u_1 \tok v_1,\ldots, u_i \tok v_i]$ for each $i$. Then we can determine $\Tk^D_{\textrm{hf}}(w)$ inductively in the following way: first calculate $\Tk^{D_1}_{\textrm{hf}}(w)$, then assuming we know $\Tk^{D_i}_{\textrm{hf}}$ we can simply apply $u_{i+1}\tok v_{i+1}$ to $\Tk^{D_i}_{\textrm{hf}}(w)$ wherever possible (working left to right) to obtain $\Tk^{D_{i+1}}_{\textrm{hf}}$. This procedure is correct by Lemma~\ref{lemma:hf-priority-order}.
We refer to $\Tk^{D_i}(w)$ as a tokenization at level $i$, and note that $\Tk^{D_i}(w)$ is a refinement of $\Tk^D(w)$.

Let $\phi\tok\tau = w_1\tok\ldots\tok w_k$ and $w=\pi(\phi \tok \tau) \pi(\psi)$. We show that the tokenization of $w$ at level $i$ is a refinement of a tokenization of the form $w_1\tok\ldots\tok w_{k-i}\tok \psi_i$, thus that $\Tk^{D_i}_{\textrm{hf}}(w)$ is a refinement of a tokenization of the form $w_1\tok\ldots w_{k-i}\tok\psi_i$, for some tokenization $\psi_i$. 
More precisely, when we use a dictionary with $i$ rules, then tokenizing $\pi(\phi\tok\tau)\pi(\psi)$ changes at most the rightmost $i$ tokens in $\phi\tok\tau$.
This implies that when $i=n$, we obtain that $\Tk^{D_n}_{\textrm{hf}}(w)$ is a tokenization of the form $w_1\tok\ldots\tok w_{k-n}\tok \psi'$, i.e.\ not only a refinement of a tokenization of the given form.


First, we show that  $\Tk^{D_1}_{\textrm{hf}}(w)$ is a refinement of a tokenization of the form $w_1\tok\ldots \tok w_{k-1}\tok\psi_1$. Note that $u_1\tok v_1$ could potentially be applied, one or more times, to the substring $w_k\pi(\psi)$, when tokenizing $\pi(\phi\tok\tau)\pi(\psi)$, but $u_1\tok v_1$ is not applied across any of the $(k-1)$ tokenization boundaries in $w_1\tok\ldots\tok w_k$. This must be the case, otherwise $u_1 \tok v_1$ 
would have been applied when computing the tokenization $w_1 \tok \ldots \tok w_k$ with the full dictionary $D$. \mbecomment{In principle we'd like to invoke Corollary~\ref{cor:bordered-stays-same} here to make it clear that not breaking the boundaries means no other changes, but it is currently only stated for SentencePiece. We can invoke Lemma~\ref{lemma:proper-unifies-semantics} to get around that, but probably Corollary~\ref{cor:bordered-stays-same} should cover both cases. Same for case below.}

Moving on to the next rule in terms of priority, $u_2\tok v_2$, we repeat the argument we used for $u_1\tok v_1$. More precisely, $u_2\tok v_2$ could potentially be applied, one or more times, to the substring $w_{k-1}\pi(\psi_1)$, when tokenizing $\pi(w_1\tok\ldots\tok w_{k-1}\tok\psi_1)$, but $u_2\tok v_2$ is not applied across any of the $(k-2)$ tokenization boundaries in $w_1\tok\ldots\tok w_{k-1}$, otherwise it would have when $\pi(\phi\tok\tau)$ was tokenized using $D$. Thus, $\Tk^{D_2}_{\textrm{hf}}(w)$ is a refinement of a tokenization of the form $w_1\tok\ldots \tok w_{k-2}\tok\psi_2$.

Simply iterate this procedure for $i=3,\ldots,n$ and the statement of the theorem follows. 
\end{proof}

\begin{corollary}
  \label{cor:left-ripple-bound-sp}%
  Let $D$ be proper, $|\tau|\ge |D|$ and $\Tk^D(\pi(\phi)\pi(\tau))=\phi\tok \tau$. Then $\Tk^D(\pi(\phi \tok \tau) \pi(\psi)) = \phi \tok \psi'$, for some tokenization $\psi'$. 
\end{corollary}
\begin{proof}
  That is, Theorem~\ref{thm:left-ripple-bound} holds for the semantics of Definition~\ref{defn:tokenization-correct} as well, since $\Tk^D_{\textrm{hf}}(w)=\Tk^D(w)$ for all $w$ when $D$ is proper, by Lemma~\ref{lemma:proper-unifies-semantics}. 
\end{proof}

\begin{example}
In this example, we explore the previous theorem.
Let $\Sigma$ be $\{a,b,c,d,e,f,g,h,i\}$ and $D$ be $[h\tok i,g\tok h,f\tok g,e\tok f,d\tok e,c\tok d,b\tok c,a\tok b]$. Also, let $\phi$ be the tokenization of the string of length $0$, $\tau = ab\tok cd\tok ef\tok gh$ and $\psi = i$. Then, $\Tk^D_{\textrm{hf}}(\pi(\phi\tok \tau)\pi(\psi))=\Tk^D_{\textrm{hf}}(abcdefghi)=a\tok bc\tok de\tok fg\tok hi$. Thus, it is not possible to shorten $\tau$ and move a prefix of $\tau$ to $\phi$.

If we shorten $D$ to be $[h\tok i]$, and make $\phi=g$, $\tau = h$ and $\psi =i$, we obtain that $\Tk^D_{\textrm{hf}}(\pi(\phi\tok\tau)\pi(\psi))= g \tok hi$. Thus, in this case, we have that the bound $|\tau|\ge |D|=1$ in the previous theorem is tight, i.e. we can not shorten $\tau$ and make it the tokenization of the empty string, and $\phi=h\tok i$, if we want $\Tk^D(\pi(\phi \tok \tau) \pi(\psi)) = \phi \tok \psi'$ in this case.
On the other hand, when taking $D=[h\tok i,f\tok g]$, then $\Tk^D_{\textrm{hf}}(\pi(fg\tok h)\pi(i))=fg\tok hi$, thus here it is sufficient to take $|\tau|=1$, i.e.\ $\tau = h$, instead of all of $fg\tok h$ (note, $|D|=2$).
\end{example}

Theorem~\ref{thm:left-ripple-bound} and Corollary~\ref{cor:left-ripple-bound-sp} may at first appear quite abstract, but they demonstrate a fact that is very useful in practice: a finite lookahead is sufficient to tokenize a string from left to right.
\begin{definition}
  The \emph{sufficient lookahead} of a proper dictionary $D$ is $(|D|+1)\cdot\max\{|uv| \mid u\tok v \in D\}$.
\end{definition}
That is, for a proper dictionary $D$ and a string $w$, if we know that $\phi$ is a prefix of $\Tk^D(w)$ (beginning the process by taking $|\phi|=0$) we can compute a prefix $\phi\tok u$ by inspecting only the sufficient lookahead many next symbols.

a finite prefix of the string, specifically the following.
That is, the symbols in $w$. Observe that this lookahead length does not depend on $|w|$. This has potential to improve tokenization performance by cache locality (where e.g.~SentencePiece~\cite{sentencepiece} and HuggingFace~\cite{huggingface-gpt-2-py} access the string contents with random access), but also enables doing streaming tokenization using a constant amount of memory, for when the entire string is not available or impractical to hold in memory. One way to express this finite state tokenization approach is as a deterministic string-to-string transducer. First, to avoid special cases for the end of the string lets define a simple normal form.
\begin{definition}
  For a dictionary $D$ over the alphabet $\sig$, let $k$ be the sufficient lookahead for $D$, assume $z\notin \sig$, then a string $v \in (\sig\cup\{z\})^*$ is the \emph{end-padding of $w$} if it is of the form $wzz\cdots z$ where there are a total of $k$ trailing $z$s.
\end{definition}
\begin{remark}
  Observe that if $v$ is $w$ end-padded then we have $\Tk^D(v)=\Tk^D(w) \tok \phi$ where $\phi=z\tok \cdots \tok z$, since $D$ contains no rules involving $z$. This means that tokenizing $v$ from the left to right we can end the procedure the moment the lookahead consists only of $z$s, as that is the padding which will always tokenize to $\phi=z\tok \cdots \tok z$.
\end{remark}
This allows us to state a straightforward left-to-right tokenization algorithm without having special cases for when the 
\begin{algorithm}
  \label{algo:finite-state}%
  Let $D$ be a proper dictionary and $k$ its sufficient lookahead. Then precompute $f : \sig^k \to \sig^{\tok}$ such that $f(w)=u$ where $\Tk^D(w)=u \tok \tau$ for some $\tau$. Observe that then $\Tk^D(ww')=u \tok \tau'$ for all $w'$ as well, by Theorem~\ref{thm:left-ripple-bound}, or rather Corollary~\ref{cor:left-ripple-bound-sp}.

  Then for any string $w$ let $v$ be its end-padding, we can then compute $\Tk^D(v)$ using the following steps.
  \begin{enumerate}
    \item Split $v=v'v''$ such that $|v'|=k$.
    \item If $v'=z\cdots z$, halt.
    \item Lookup $f(v')=u$, output $u$.
    \item Split $v=uu'$ (observe that $u$ must be a prefix of $v'$ and in turn $v$).
    \item Update $v$ to be $u'$, then go to 1.
  \end{enumerate}
\end{algorithm}

\begin{theorem}
  For any fixed proper $D$ and any string $w$ Algorithm~\ref{algo:finite-state} outputs $\Tk^D(w)$ using space $\mathcal{O}(1)$ and time $\mathcal{O}(|w|)$.
\end{theorem}
\begin{proof}
\end{proof}

\bibliographystyle{eptcs}
\bibliography{main}

\end{document}

%% file: defs.tex
\newcommand{\eqed}{\hfill\ensuremath{\diamond}}
\newcommand{\eps}{\ensuremath{\varepsilon}}

\newcommand{\sig}{\ensuremath{\Sigma}}

\newcommand{\false}{\textit{False}}
\newcommand{\height}{\ensuremath{\textit{ht}}}

\makeatletter
\long\def\ifnodedefined#1#2#3{%
    \@ifundefined{pgf@sh@ns@#1}{#3}{#2}%
}
\makeatother

\newcommand{\xcomment}[2]{%
    \tikz[overlay,remember picture] {%
        \node[anchor=north west,xshift=.75em,yshift=.75\baselineskip] (c) at (0,0-|current page text area.east) {\parbox{3cm}{\scriptsize\raggedright #1}};%
        \draw (c.north west) edge[line width=3pt,#2!50!white] (c.south west);%
        \draw[#2!50!white] (0,-.3em-|c.west) -- (0,-.3em);%
        \draw[#2!50!white] (0,-.3em) edge[->] (0,.2em);%
    }%
}

\newcommand{\mbecomment}[1]{\xcomment{#1}{blue}}
\newcommand{\brinkcomment}[1]{\xcomment{#1}{red}}

\usepackage{amsmath}
\usepackage{amssymb}
\usepackage{slashed}
\usepackage{wrapfig}
\usepackage{MnSymbol} 
\usepackage{stmaryrd}
\usepackage{mathtools}
\usepackage{moreverb}
\usepackage{amsthm}

\usepackage{tikz}
\usepackage{tikzpagenodes}
\usetikzlibrary{positioning,automata,shapes,calc,patterns,patterns.meta,decorations.pathreplacing,arrows,arrows.meta}

\tikzset{>=latex'}

\providecommand{\keywords}[1]
{{
  \small	
  \textbf{\textit{Keywords---}} #1
}}

\newcommand{\change}[1]{\textcolor{black}{#1}}
\newcommand{\changeB}[1]{\textcolor{black}{#1}}
\renewcommand{\emptyset}{\varnothing}

\pgfdeclarelayer{background}
\pgfdeclarelayer{doublebackground}
\pgfdeclarelayer{triplebackground}
\pgfsetlayers{triplebackground,doublebackground,background,main}

\makeatletter
\newcommand\footnoteref[1]{\protected@xdef\@thefnmark{\ref{#1}}\@footnotemark}
\makeatother